\documentclass[AMA,Times1COL]{WileyNJDv5}

\usepackage{algpseudocode}
\usepackage{algorithm}
\usepackage{amssymb} 
\usepackage{booktabs} 
\usepackage{amsmath} 
\usepackage{xspace} 
\usepackage{url} 
\usepackage{subcaption} 
\usepackage{xcolor}

\newcommand{\thp}{\tau^{HP}}
\newcommand{\tpc}{\tau^{PC}}
\newcommand{\thpi}{\tau^{HP}_i}

\newcommand{\tpci}{\tau^{PC}_i}

\newcommand{\shi}{\sigma^H_i}

\newcommand{\spi}{\sigma^P_i}

\newcommand{\Zhpi}{Z^{HP}_i}

\newcommand{\ZIi}{Z^I_i}
\newcommand{\fI}{f_\textsf{I}}
\newcommand{\fH}{f_\textsf{H}}
\newcommand{\SH}{S_\textsf{H}}

\newcommand{\fprog}{f_{\textsf{prog}}}
\newcommand{\Sprog}{S_{\textsf{prog}}}

\newcommand{\Exp}{\textsf{Exp}}
\newcommand{\Wei}{\textsf{Wei}}

\newcommand{\Ber}{\textsf{Ber}}
\newcommand{\Ga}{\textsf{Ga}}
\newcommand{\Be}{\textsf{Be}}

\newcommand{\trfh}{\textsf{Trunc-F}_\textsf{H}}
\newcommand{\FH}{\textsf{F}_\textsf{H}}
\newcommand{\FP}{\textsf{F}_\textsf{P}}
\newcommand{\F}{\textsf{F}}
\newcommand{\Fprog}{\textsf{F}_\textsf{prog}}

\newcommand{\muprog}{\mu_\textsf{prog}}

\newcommand{\Unif}{\textsf{U}}
\newcommand{\iton}{i=1,\dots,n}

\newcommand{\lH}{\lambda_\textsf{H}}
\newcommand{\lprog}{\lambda_{\textsf{prog}}}
\newcommand{\aH}{\alpha_\textsf{H}}
\newcommand{\aprog}{\alpha_{\textsf{prog}}}

\newcommand{\Y}{\mathbf{Y}}
\newcommand{\Z}{\mathbf{Z}}
\newcommand{\X}{\mathbf{X}}
\newcommand{\x}{\mathbf{x}}
\newcommand{\BCSC}{Breast Cancer Surveillance Consortium\xspace}
\newcommand{\MCMC}{Markov chain Monte Carlo\xspace}

\usepackage{tikz}
\usetikzlibrary{
	arrows.meta,
	calc,
	datavisualization.formats.functions,
	decorations.text,
	plotmarks,
	decorations.markings,
	decorations.pathreplacing,
        trees,
        shapes.multipart}
\tikzset{ shorten <>/.style={ shorten >=#1, shorten <=#1 } }

\articletype{Article Type}%

\received{Date Month Year}
\revised{Date Month Year}
\accepted{Date Month Year}
\journal{Statistics in Medicine}
\volume{00}
\copyyear{2024}
\startpage{1}

\begin{document}

\title{A Bayesian Approach for Fitting Semi-Markov Mixture Models of Cancer Latency to Individual-level Data}

\author[1]{Rapha\"{e}l N. Morsomme}

\author[2]{Shannon T. Holloway}

\author[2,3]{Marc D. Ryser}

\author[1]{Jason Xu}

\authormark{MORSOMME \textsc{et al.}}
\titlemark{A Bayesian Approach for Fitting Semi-Markov Mixture Models of Cancer Latency to Individual-level Data}

\address[1]{\orgdiv{Department of Statistical Science}, \orgname{Duke University}, \orgaddress{\state{North Carolina}, \country{United States}}}

\address[2]{\orgdiv{Department of Population Health Sciences}, \orgname{Duke University}, \orgaddress{\state{North Carolina}, \country{United States}}}

\address[3]{\orgdiv{Department of Mathematics}, \orgname{Duke University}, \orgaddress{\state{North Carolina}, \country{United States}}}

\corres{Corresponding authors Jason Xu and Marc D. Ryser (joint senior authors). \email{jqxu@ucla.edu, marc.ryser@duke.edu }}

\presentaddress{This is sample for present address text this is sample for present address text.}


\abstract[Abstract]{
Multi-state models of cancer natural history are widely used for designing and evaluating cancer early detection strategies. Calibrating such models against longitudinal data from screened cohorts is challenging, especially when fitting non-Markovian mixture models against individual-level data. Here, we consider a family of semi-Markov mixture models of cancer natural history and introduce an efficient data-augmented Markov chain Monte Carlo sampling algorithm for fitting these models to individual-level screening and cancer diagnosis histories. Our fully Bayesian approach supports rigorous uncertainty quantification and model selection through leave-one-out cross-validation, and it enables the estimation of screening-related overdiagnosis rates. We demonstrate the effectiveness of our approach using simulated data, showing that the sampling algorithm efficiently explores the joint posterior distribution of model parameters and latent variables. Finally, we apply our method to data from the US Breast Cancer Surveillance Consortium and estimate the extent of breast cancer overdiagnosis associated with mammography screening. The sampler and model comparison method are available in the \texttt{R} package \texttt{baclava}.
}


\jnlcitation{\cname{%
\author{Taylor M.},
\author{Lauritzen P},
\author{Erath C}, and
\author{Mittal R}}.
\ctitle{On simplifying ‘incremental remap’-based transport schemes.} \cjournal{\it J Comput Phys.} \cvol{2021;00(00):1--18}.}

\maketitle

\renewcommand\thefootnote{}
\footnotetext{\textbf{Abbreviations:} MCMC, Markov chain Monte Carlo; MH, Metropolis-Hastings; DA-MCMC, data-augmented Markov chain Monte Carlo; BCSC, Breast Cancer Surveillance Consortium; CV, cross-validation; ALOOCV, approximate leave-one-out cross-validation.}

\renewcommand\thefootnote{\fnsymbol{footnote}}
\setcounter{footnote}{1}

\section{Introduction}
Cancer screening aims to detect early markers of cancer before symptoms arise, facilitating timely intervention for improved patient outcomes. The efficacy of screening depends on the duration of the pre-clinical phase, during which a tumor is asymptomatic but screen-detectable. If the tumor latency period is too short, even frequent screening does not significantly advance the time of diagnosis. Conversely, a long latency period can result in \textit{overdiagnosis} and \textit{overtreatment}---the diagnosis and unnecessary treatment of tumors that would not have caused symptoms or other harm in the person's remaining lifetime \citep{welch2010overdiagnosis, duffy2013overdiagnosis}. Understanding the duration of tumor latency is thus crucial for assessing the benefits and harms of interventions that follow positive screening tests.

Estimating tumor latency from cancer incidence data is challenging because the onset of pre-clinical disease in healthy individuals is not directly observable. In unscreened populations, only the time of clinical diagnosis, which coincides with the end of the tumor latency period, is known. In this case, we cannot distinguish the time intervals before and after pre-clinical onset. In other words, the period of tumor latency is not identifiable. Fortunately, screening methods that utilize biological or imaging markers offer a selective insight into the latent disease state, thus enabling the identification of tumor latency.

There is a rich history on estimating cancer latency from screening data using multi-state models that dates back to the 1980s. These models range from simple three-state models \citep{day1984simplified, shen1999parametric}---in which individuals are \textit{healthy} (disease-free), \textit{pre-clinical} (with an asymptomatic screen-detectable cancer) or \textit{clinical} (with a symptomatic cancer)---to more complex models that seek to emulate the biological complexity of cancer progression \citep{rutter2010evidence, seigneurin2011overdiagnosis}. Typically, the underlying stochastic processes of such multi-state models are assumed to be Markovian for mathematical tractability \citep{lange2015joint}. Under this assumption, the unobserved onset times can be integrated out analytically, giving a closed-form expression for the likelihood of  observed data. This enables maximum likelihood estimation \citep{olsen2006overdiagnosis, wu2018overdiagnosis} and Bayesian inference \citep{duffy2005overdiagnosis, ryser2022estimation}. However, relaxing the Markov assumption often leads to a departure from the simpler mathematical structures, and more complex non-Markovian models are usually calibrated using simulation-based or ``likelihood-free'' approximation methods, such as approximate Bayesian computation \citep{rutter2009bayesian, alagoz2018university, bondi2023approximate} or pseudo-likelihood methods \citep{moore2001pseudo}.

In this paper, we present a method for exact Bayesian inference of cancer natural history using a flexible family of semi-Markov mixture models that account for both progressive and non-progressive (indolent) disease states. The proposed approach accommodates individual-level screening and cancer diagnosis histories, and, because it does not rely on model-based simulation, is scalable to large cohorts of screened individuals. By relaxing the Markov assumption, we may then flexibly specify the hazard function, which in turn leads to improved model realism \citep{hsieh2002assessing, cheung2022multistate}. Since the semi-Markov models do not admit a closed form expression for the observed data likelihood, we conduct inference using data augmentation \citep{tanner1987calculation}. This strategy augments the observed data with latent variables in such a way that the resulting \textit{complete-data} likelihood is amenable to iterative sampling. Executing this idea in the context of cancer natural history modeling is nontrivial because the number of latent variables is proportional to the number of individuals in the study, and cancer screening cohorts often comprise tens or hundreds of thousands of individuals. The resulting joint posterior distribution is complex and high-dimensional, rendering standard rejection or importance sampling approaches impractical. Moreover, gradient-based approaches, such as Hamiltonian Monte Carlo, which are well-suited for high-dimensional problems \citep{girolami2011riemann}, are not applicable to our family of models due to the presence of discrete latent variables, and of discontinuities in the likelihood from the discrete screening outcomes.

Instead, our DA-MCMC sampler efficiently explores the joint posterior distribution of the model parameters together with the latent variables, enjoying sound theoretical guarantees. This fully Bayesian approach automatically accounts for the uncertainty in the unobserved onset time of cancers, and allows the estimation of any relevant function of the parameters, e.g.\ the overdiagnosis rate, through its posterior distribution. Moreover, because the complete-data likelihood has a tractable closed form, we can rigorously compare models using an extension of ALOOCV \citep{gelfand1992model, vehtari2017practical} to latent models.

The paper is structured as follows. In Section~\ref{sec:model}, we introduce the underlying family of semi-Markov mixture models of cancer natural history; in Section~\ref{sec:method}, we describe the inferential framework; and, in Section~\ref{sec:mcmc}, we detail the construction of the DA-MCMC algorithm. We examine the mixing properties of the sampler through simulation studies in Section~\ref{sec:sim} and apply it to estimate the extent of overdiagnosis in a real-world cohort of women undergoing mammography screening in Section~\ref{sec:BCSC}. The DA-MCMC sampler is publicly available in the \verb|R| package implementation \verb|baclava| on CRAN, and \verb|R| code for reproducing the results present in this paper is available on GitHub (\url{https://github.com/rmorsomme/baclava-manuscript}).

\section{Semi-Markov mixture model of cancer progression} \label{sec:model}

\begin{figure}
\centering
    \begin{subfigure}{\textwidth}   
    \centering
        \begin{tikzpicture}[every node/.style={scale=1}, every text node part/.style={align=center}]
            \draw (0.0, 0.0) node[align=center, draw, rectangle, minimum width = 2.5cm, minimum height=0.75cm] (1) {Healthy ($H$)};
            \draw (4.0, 0.0) node[align=center, draw, rectangle, minimum width = 2.5cm, minimum height=2.75cm] (2) {
            Pre-clinical \\ cancer ($P$) \\ \\
            \\
            \\
            };
            \draw (8.0, 0.0) node[align=center, draw, rectangle, minimum width =2.5cm, minimum height=1.2cm] (4) {Clinical \\ cancer ($C$)};

            \draw (4, 0.0) node[align=center, draw, rectangle, minimum width = 2cm, minimum height=0.5cm] (3) {progressive};
            \draw (4, -0.8) node[align=center, draw, rectangle, minimum width = 2cm, minimum height=0.5cm] {indolent};
        
        \draw [thick,->,shorten >=1mm] (1) to (2) node[midway,above] {};
        \draw [thick,->,shorten >=1mm] (3) to (4) node[midway,above] {};
        
    \end{tikzpicture}
    \caption{Compartmental model for cancer progression with a mixture of progressive and indolent pre-clinical cancers.}
    \label{fig:model-a}
\end{subfigure}
\begin{subfigure}{\textwidth}
    \centering
    \begin{tikzpicture}
        \definecolor{cbgray1}{gray}{0.5}
        \definecolor{cbgray2}{gray}{0.25}    
        
        \draw [thick] (0.0,0.0) -- (0.0,2.75);
        \draw[decoration={markings,mark=at position 1 with {\arrow[scale=2,>=stealth]{>}}},postaction={decorate},thick] (0.0,0.0) -- (7.75,0.0);
        \node at (7.2,-0.45) {Age};

        \draw (0.5, -0.1) -- (0.5, 0.1);
        \node at (0.5, -0.4) {$t_0$};

        \draw (0.0, 2.25) node[anchor = east] {($H$)};
        \draw (0.0, 1.5) node[anchor = east] {($P$)};
        \draw (0.0, 0.75) node[anchor = east] {($C$)};

        \draw (0.5, 2.25) -- (6.0, 2.25);
        \draw (6.0, 2.25) node[anchor = west] {Individual 1};

        \draw [thick,densely dotted] (0.5, 2.15) -- (2.0, 2.15);
        \draw [thick,densely dotted] (2.0, 2.15) -- (2.0, 1.4);
        \draw [thick,densely dotted] (6.0, 1.4) -- (2.0, 1.4);
        \draw (6.0, 1.4) node[anchor = west] {Individual 2};
        \draw (2.0, -0.1) -- (2.0, 0.1);
        \node at (2.0, -0.4) {$\tau^{HP}_2$};

        \draw [dashed] (0.5, 2.05) -- (3.0, 2.05);
        \draw [dashed] (3.0, 2.05) -- (3.0, 1.3);
        \draw [dashed] (3.0, 1.3) -- (4.5, 1.3);
        \draw [dashed] (4.5, 1.3) -- (4.5, 0.55);
        \draw (4.5, 0.55) node[anchor = west] {Individual 3};
        \draw (3.0, -0.1) -- (3.0, 0.1);
        \node at (3.0, -0.4) {$\tau^{HP}_3$};
        \draw (4.5, -0.1) -- (4.5, 0.1);
        \node at (4.5, -0.4) {$\tau^{PC}_3$};
        
    \end{tikzpicture}
    \caption{Illustration of possible individual clinical trajectories through the compartmental model.
    }
    \label{fig:model-b}
\end{subfigure}        
\caption{
    Panel (a) depicts the compartmental model for the progression of cancer. Panel (b) illustrates possible individual trajectories through the compartmental model. Individuals start in compartment $H$. Only transitions from $H$ to $P$  and from $P$ to $C$ are possible, with compartment $P$ being absorbing for indolent individuals and compartment $C$ being absorbing for all others. In Panel (b), individual 1 does not develop pre-clinical cancer before the end of the observation period and, therefore, remains in $H$ until that time.
    Individual 2 develops pre-clinical cancer at age $\tau^{HP}_2$, at which time they transition to $P$, but they do not develop clinical cancer before the end of the observation period and, therefore, remain in $P$. Note that the pre-clinical cancer of individual 2 may be either indolent, in which case $\tau^{PC}_2=\infty$, or progressive with a clinical onset time after the end of the observation period.
    Individual 3 develops pre-clinical cancer at age $\tau^{HP}_3$ and transitions to state $P$. At age $\tau^{PC}_3$, they develop clinical cancer and thus transition to $C$. As $C$ is absorbing, the observation period ends with their transition time into $C$.
}
\label{fig:model}
\end{figure}

We model the dynamics of cancer progression as a multi-state mixture model. The model, depicted in Figure~\ref{fig:model-a},  consists of the three compartments \textit{healthy} ($H$), \textit{pre-clinical cancer} ($P$) and \textit{clinical cancer} ($C$). Following \citep{ryser2019identification}, pre-clinical cancers are modeled as a mixture of progressive and indolent tumors. All individuals start in the healthy state $H$, and transition to a state of pre-clinical screen-detectable cancer $P$ after a randomly distributed time. Once in $P$, individuals with an indolent cancer stay in this compartment indefinitely, while individuals with a progressive cancer transition to $C$ when their cancer becomes symptomatic.

Several early works consider purely progressive models, in which all pre-clinical tumors progress to clinical disease after a finite time \citep{day1984simplified,shen1999parametric,shen2001screening}. Allowing a proportion of pre-clinical cancers to be non-progressive is particularly important for slowly growing or indolent disease entities, such as low-risk prostate cancers or in situ breast cancers \citep{baker2014lead}. Despite the challenges of fitting such mixture models in practice \citep{ryser2019identification}, they have gained in popularity, especially in the study of breast cancer progression  \citep{chen1996markov,duffy2005overdiagnosis,olsen2006overdiagnosis,seigneurin2011overdiagnosis, alagoz2018university, wu2018overdiagnosis,ryser2022estimation}.

To formalize the model structure, we let the starting age $t_0\ge0$ correspond to the age at which individuals become susceptible to developing pre-clinical cancer, so that the hazard rate of pre-clinical cancer is assumed to be zero before $t_0$. Next, let $\thp$ be the age at pre-clinical onset (the transition from $H$ to $P$) and let $\tpc$ be the age at clinical onset (the transition from $P$ to $C$); by definition, $t_0 \le \thp \le \tpc$. When a tumor is indolent, we set $\tpc=\infty$, reflecting the fact that indolent cancers never progress to the clinical stage. Figure~\ref{fig:model-b} illustrates the evolution of three individuals throughout this multi-state model. 

The associated stochastic process is characterized by the probabilistic mechanisms governing the transitions between the compartments. Denote the waiting time until pre-clinical cancer with $\sigma^H =\thp-t_0$ and the pre-clinical sojourn time with $\sigma^P =\tpc-\thp$, corresponding to the interval between the onset of pre-clinical cancer and the symptomatic diagnosis. An individual with an indolent tumor has $\sigma^P=\infty$. We assume that 
\begin{equation}  \label{eq:F}
    \sigma^H \sim \FH, \quad \sigma^P \sim \FP,
\end{equation}
independently, where $\FH$ is a distribution on the positive line, and $\FP$ corresponds to the mixture cure model \citep{peng2014cure},
\begin{equation} \label{eq:F-p}
    \FP = \psi \delta_\infty + (1-\psi)\Fprog.
\end{equation}
with a point mass at $\infty$ for indolent cancers and a density on the positive line for progressive cancers. Here, $\psi \in [0,1]$ is the proportion of indolent cancers and $\delta_\infty$ the Dirac distribution with its mass at $\infty$. Distribution~\eqref{eq:F-p} implies that a proportion $\psi$ of tumors are indolent with infinite sojourn times while the sojourn time of progressive tumors follows distribution $\Fprog$.

Our inferential framework will apply to arbitrary choices of parametric distributions on the positive line for $\FH$ and $\Fprog$, resulting in a flexible semi-Markovian formulation. We write $\theta$ to refer to the collective parameters of $\FH$ and $\Fprog$. This formulation encompasses existing models as special cases:  for example, letting $\FH$ and $\Fprog$ be exponential distributions results in the Markov mixture model of \citep{seigneurin2011overdiagnosis}, and $\psi=0$ recovers the purely progressive model of \citep{shen1999parametric, shen2001screening}. Introducing the indolence indicator $I$, which equals $1$ if the cancer is indolent and $0$ otherwise, we can also express the distribution \eqref{eq:F-p} hierarchically as
\begin{align} \label{eq:I}
    I \sim \Ber(\psi), \quad 
    \sigma^P\mid I \sim 
    \begin{cases}
    \delta_\infty, &\quad I = 1 \\
    \Fprog, &\quad I = 0
    \end{cases},
\end{align}
where $\Ber(p)$ denotes the Bernoulli distribution with probability parameter $p$.

\section{Likelihood and inferential methodology}  \label{sec:method}

\subsection{Likelihood of complete data}  \label{sec:likelihood}

Consider a sample of $n$ independent individuals, and let the subscript $\iton$ refer to variables of the $i$th individual. 
We refer to the data for individual $i$ as $X_i=\left(I_i, \thpi, \tpci\right)$ and write $\X=(X_1, \dots, X_n)$ for the data of the entire population. By independence between the individuals, the likelihood of the process described by equations \eqref{eq:F}-\eqref{eq:I} factorizes as
\begin{equation}  \label{eq:likelihood-complete}
    \tilde{L}(\theta, \psi;\X)= \prod_{i=1}^n \tilde{L}_i(\theta, \psi;X_i), \quad \text{where} \quad \tilde{L}_i(\theta, \psi;X_i) = \begin{cases}
        \SH\left( c_i - t_0; \theta\right), & \quad c_i < \thpi\\
        \fH\left( \shi; \theta\right) \fI(I_i; \psi) \Sprog\left( c_i - \thpi; \theta\right)^{1-I_i}, & \quad \thpi \le c_i < \tpci \\
        \fH\left( \shi; \theta\right)(1-\psi)\fprog\left( \spi; \theta\right), & \quad \tpci \le c_i
    \end{cases}.
\end{equation}
is the likelihood of person $i$ observed until their censoring age $c_i$. Here, $f_\cdot$ and $S_\cdot$ are the density and survivial function of the corresponding distribution $\F_\cdot$ and $\fI$ is the probability mass function of the Bernoulli distribution~\eqref{eq:I}.

We can understand the individual likelihood \eqref{eq:likelihood-complete} as follows.
The contribution of a person $i$ to the likelihood $\tilde{L}(\theta;\X)$ depends on which compartment they are in at age $c_i$.
If individual $i$ does not develop pre-clinical cancer by the censoring age $\left( c_i < \thpi\right)$, then their contribution is the survival function $\SH\left(c_i - t_0; \theta\right)$ of $\FH$. 
If they develop pre-clinical cancer but not clinical cancer by the censoring time $\left(\thpi\le c_i<\tpci\right)$, then their contribution consists of the product of the density $\fH\left(\shi; \theta\right)$ of $\FH$, the Bernoulli probability mass function of the indolence indicator $\fI(I_i; \psi)$, and $\Sprog\left(c_i - \thpi; \theta\right)^{1-I_i}$ which is either the constant function $1$ for indolent cases $(I_i=1)$ or the survival function of $\Fprog$ for progressive cases ($I_i=0$). 
Finally, if the individual develops clinical cancer by the censoring time $\left(\tpci \le c_i\right)$, then their contribution is the product of $\fH\left(\shi; \theta\right)$, the probability $1-\psi$ of a progressive cancer, and the density $\fprog\left(\spi; \theta\right)$ of distribution $\Fprog$.

\subsection{Partially observed data} \label{sec:observed-data}

\begin{figure}
\centering
    \begin{subfigure}[b]{0.3\textwidth}
        \centering
        \begin{tikzpicture}
            \definecolor{cbgray1}{gray}{0.5}
            \definecolor{cbgray2}{gray}{0.25}    
            
            \draw [thick] (0.0,0.0) -- (0.4,0.0);
            \draw [thick] (0.3,-0.15) -- (0.5,0.15);
            \draw [thick] (0.5,-0.15) -- (0.7,0.15);
            \draw[decoration={markings,mark=at position 1 with {\arrow[scale=2,>=stealth]{>}}},postaction={decorate},thick] (0.6,0.0) -- (4.0,0.0);
            \node at (3.5,-0.65) {\begin{tabular}{c} Age \\ (in years) \end{tabular}};

    
            \draw (0.0, 2.25) node[anchor = east] {($H$)};
            \draw (0.0, 1.5) node[anchor = east] {($P$)};
            \draw (0.0, 0.75) node[anchor = east] {($C$)};

            \draw (0.0, -0.1) -- (0.0, 0.1);
            \node at (0.0, -0.4) {$t_0$};
    
            \draw [dotted, color=cbgray2] (1.0, 0.0) -- (1.0, 2.75);
            \draw (1.0, -0.1) -- (1.0, 0.1);
            \node at (1.0, -0.4) {$40$};
            \node at (1.0, 3.0) {$O^1_1=0$};
            
            \draw [dotted, color=cbgray2] (2.5, 0.0) -- (2.5, 2.75);
            \draw (2.5, -0.1) -- (2.5, 0.1);
            \node at (2.5, -0.4) {$45$};
            \node at (2.5, 3.0) {$O^1_2=0$};

            \node at (3.0, 0.0)[circle,fill,inner sep=1.5pt]{};
            \node at (3.0, 0.35) {$46.5$};
    
            \draw (0.0, 2.25) -- (3.0, 2.25);
            
        \end{tikzpicture}
        \caption{Censored individual.}
        \label{fig:observed-data-1}
    \end{subfigure}
    \hfill
    \begin{subfigure}[b]{0.3\textwidth}
        \centering
        \begin{tikzpicture}
            \definecolor{cbgray1}{gray}{0.5}
            \definecolor{cbgray2}{gray}{0.25}
            
            \draw [thick] (0.0,0.0) -- (0.4,0.0);
            \draw [thick] (0.3,-0.15) -- (0.5,0.15);
            \draw [thick] (0.5,-0.15) -- (0.7,0.15);
            \draw[decoration={markings,mark=at position 1 with {\arrow[scale=2,>=stealth]{>}}},postaction={decorate},thick] (0.6,0.0) -- (4.0,0.0);
            \node at (3.5,-0.65) {\begin{tabular}{c} Age \\ (in years) \end{tabular}};
    
            \draw (0.0, 2.25) node[anchor = east] {($H$)};
            \draw (0.0, 1.5) node[anchor = east] {($P$)};
            \draw (0.0, 0.75) node[anchor = east] {($C$)};

            \draw (0.0, -0.1) -- (0.0, 0.1);
            \node at (0.0, -0.4) {$t_0$};
    
            \draw [dotted, color=cbgray2] (1.0, 0.0) -- (1.0, 2.75);
            \draw (1.0, -0.1) -- (1.0, 0.1);
            \node at (1.0, -0.4) {$40$};
            \node at (1.0, 3.0) {$O^2_1=0$};
            
            \draw [dotted, color=cbgray2] (2.5, 0.0) -- (2.5, 2.75);
            \draw (2.5, -0.1) -- (2.5, 0.1);
            \node at (2.5, -0.4) {$45$};
            \node at (2.5, 3.0) {$O^2_2=1$};
            
            \draw (0.0, 2.25) -- (0.5, 2.25);
            \draw (0.5, 2.25) -- (0.5, 1.5);
            \draw (0.5, 1.5) -- (2.5, 1.5);

            
        \end{tikzpicture}
        \caption{Screen-detected individual.}
        \label{fig:observed-data-2}
    \end{subfigure}
    \hfill
    \begin{subfigure}[b]{0.3\textwidth}
        \centering
        \begin{tikzpicture}
            \definecolor{cbgray1}{gray}{0.5}
            \definecolor{cbgray2}{gray}{0.25}
            
            \draw [thick] (0.0,0.0) -- (0.4,0.0);
            \draw [thick] (0.3,-0.15) -- (0.5,0.15);
            \draw [thick] (0.5,-0.15) -- (0.7,0.15);
            \draw[decoration={markings,mark=at position 1 with {\arrow[scale=2,>=stealth]{>}}},postaction={decorate},thick] (0.6,0.0) -- (4.0,0.0);
            \node at (3.5,-0.65) {\begin{tabular}{c} Age \\ (in years) \end{tabular}};
    
            \draw (0.0, 2.25) node[anchor = east] {($H$)};
            \draw (0.0, 1.5) node[anchor = east] {($P$)};
            \draw (0.0, 0.75) node[anchor = east] {($C$)};

            \draw (0.0, -0.1) -- (0.0, 0.1);
            \node at (0.0, -0.4) {$t_0$};
    
            \draw [dotted, color=cbgray2] (1.0, 0.0) -- (1.0, 2.75);
            \draw (1.0, -0.1) -- (1.0, 0.1);
            \node at (1.0, -0.4) {$40$};
            \node at (1.0, 3.0) {$O^3_1=0$};
            
            \draw [dotted, color=cbgray2] (2.5, 0.0) -- (2.5, 2.75);
            \draw (2.5, -0.1) -- (2.5, 0.1);
            \node at (2.5, -0.4) {$45$};
            \node at (2.5, 3.0) {$O^3_2=0$};
            
            \draw (0.0, 2.25) -- (1.6, 2.25);
            \draw (1.6, 2.25) -- (1.6, 1.5);
            \draw (1.6, 1.5) -- (2.8, 1.5);
            \draw (2.8, 1.5) -- (2.8, 0.75);

            \node at (2.8, 0.0)[circle,fill,inner sep=1.5pt]{};
            \node at (2.8, 0.35) {$46$};
            
        \end{tikzpicture}
        \caption{Interval-detected individual}
        \label{fig:observed-data-3}
    \end{subfigure}
\par\bigskip
    \begin{subfigure}[!htbp]{0.4\textwidth} 
        \centering
        \begin{tabular}{ccc}
            \toprule
            Individual & Screen age & Screen outcome \\
            \midrule
            $1$ & $40$ & $0$ \\
            $1$ & $45$ & $0$ \\
            $2$ & $40$ & $0$ \\
            $2$ & $45$ & $1$ \\
            $3$ & $40$ & $0$ \\
            $3$ & $45$ & $0$ \\
            \bottomrule
        \end{tabular}
        \caption{Observed screen data for the 3 individuals corresponding to Figures~\ref{fig:observed-data-1}-\ref{fig:observed-data-3}.}
        \label{tab:observed-data-screen}
    \end{subfigure}
\hfill
    \begin{subfigure}[!htbp]{0.4\textwidth} 
        \centering
        \begin{tabular}{cccc}
            \toprule
            Individual & $t^{PC}_i$ & Censoring age \\
            \midrule
            $1$ & $46.5$ & $46.5$ \\
            $2$ & $45$ & $45$ \\
            $3$ & $46$ & --- \\
            \bottomrule
        \end{tabular}
        \caption{Observed right-censored transition age into $C$, $t^{PC}_i$, and censoring age for the 3 individuals corresponding to Figures~\ref{fig:observed-data-1}-\ref{fig:observed-data-3}.}
        \label{tab:observed-data-cancer}
    \end{subfigure}
\caption{
Observed data for the individuals depicted in Figure~\ref{fig:model-b}. 
Panels (a-c): diagrams displaying the trajectories and screen results. Panel (d): screen results. Panel (e): right-censored onset age of clinical cancer. 
In Panel (a), the individual is in $H$ at the two screen times, the two screens are therefore negative. The individual is censored at age $46.5$, $1.5$ years after their last screen.
In Panel (b), the individual is in $P$ when the two screens take place. The first screen is a false negative, and the second is positive. The individual is censored at age $45$, their age at the positive screen.
In Panel (c), the screen at age $45$ is a false negative. The onset time of the clinical cancer at age $46$ is observed.
}
\label{fig:observed-data}
\end{figure}

We now turn our attention to real-world data sources, such as prospective screening trials and observational screening cohorts, In such settings, the data consist of a series of screens and associated results, as well as clinical cancer diagnoses.  Let $n_i$ denote the number of screens available for person $i$, with their age and outcome at each of these screening times contained in the vectors $\left(t^i_1, \dots t^i_{n_i}\right)$ and $\left(O^i_1, \dots O^i_{n_i}\right) \in \{0,1\}^{n_i}$, respectively, where $t_1^i=e_i$ corresponds to the age at study entry; $O^i_j=1$ when the $j$th screen of participant $i$ is positive and $O^i_j=0$ when it is negative. These data are complemented by the time $t^{PC}_i=\min\left\lbrace c_i, \tpci\right\rbrace$, which corresponds to either the age of diagnosis with a screen-detected cancer, the age of diagnosis with a clinical cancer, or to the age of censoring otherwise, and by the indicator $\delta_i = \delta_{\tpci}\left(t^{PC}_i\right)$, which takes on the values $1$ when $t^{PC}_i=\tpci$, that is, for individuals with a clinical cancer, and $0$ otherwise.

We assume that the screening ages are non-informative, that is, that they are \textit{ignorable} or independent of the disease process and of the model parameters. Under this assumption, we may condition on them in our analysis \citep{gelman2013bayesian}. We therefore denote the observed variables of individual $i$ with $Y_i=\left(O^i_1,\dots,O^i_{n_i}, t^{PC}_i, \delta_i\right)$ and condition on the ignorable screen ages $(t^i_1,\dots,t^i_{n_i})$. We further write $\Y = (Y_1, \dots, Y_n)$. Tables~\ref{tab:observed-data-screen} and \ref{tab:observed-data-cancer} present the observed data of a sample of three individuals for illustration.

The screening outcomes are modeled as mutually independent Bernoulli random variables
\begin{equation} \label{eq:p}
O^i_j \sim \Ber\left(p^i_j\right), \quad j=1,\dots,n_i, \quad \iton, \qquad \text{where}  \quad  p^i_j = \begin{cases}
        0, \quad t^i_j < \thpi  \\
        \beta, \quad \thpi \le t^i_j < \tpci \\
    \end{cases}.
\end{equation}
Equation~\eqref{eq:p} indicates that the probability of positive screen $p^i_j$ is $0$ when a person is in the $H$ compartment and is $\beta$ when they are in $P$; the parameter $\beta\in[0,1]$ is, therefore, interpreted as the \textit{screening sensitivity} and is not assumed to be known. Upon a cancer diagnosis---whether clinical or screen-detected---an individual receives treatment and leaves the screening program, implying that any positive screen is the last screen of the individual. Consequently, $p^i_j$ need not be defined when $t^i_j > \tpci$.

The observed data partition the individuals into three groups. The \textit{censored} group consists of individuals with no positive screens and who do not develop a clinical cancer before their censoring age. These individuals either remain in $H$ until their censoring age, or transition into $P$ and remain there until their censoring age; any screens conducted when they are in $P$ are false negatives (see Figure~\ref{fig:observed-data-1}). The \textit{screen-detected} group includes individuals who are in state $P$ at the time of the positive screen (see Figure~\ref{fig:observed-data-2}). Finally, the \textit{interval-detected} group contains those with no positive screen but whose cancer reaches the clinical stage before their censoring age (see Figure~\ref{fig:observed-data-3}). These individuals may or may not have false negative screens.

\subsection{Bayesian inference via data-augmented \MCMC} \label{sec:bayesian-inference}

We fit the semi-Markov model defined by Eqs.~\eqref{eq:F}-\eqref{eq:p} to censored data $\Y$ in a Bayesian framework. By Bayes Theorem, $\pi(\theta, \psi,\beta\mid\mathbf{Y}) \propto \pi(\theta,\psi,\beta) L(\theta,\psi,\beta; \Y)$ where $\pi(\theta,\psi,\beta)$ is the prior on the parameters $(\theta,\psi,\beta)$, and the marginal likelihood
\begin{equation} \label{eq:likelihood-marginal}
    L(\theta, \psi,\beta; \Y) = \int f(\Y \mid \x, \theta, \psi, \beta) L(\theta,\psi;\x) d\x,
\end{equation}
is related to the complete-data likelihood~\eqref{eq:likelihood-complete} via integration. Equation~\eqref{eq:likelihood-marginal} integrates over all the possible clinical histories of the individuals consistent with the observed data. No general closed-form expression for this integral exists in the semi-Markov case \citep{hsieh2002assessing}, and numerical integration over this high-dimensional latent space is intractable.

To overcome this difficult integration step, DA-MCMC takes advantage of the tractable form ~\eqref{eq:likelihood-complete}. The observed data $\Y$ are augmented with the latent variables $\Z$ such that $L(\theta,\psi;\Z)$ and $f(\Y \mid\theta, \psi, \beta, \Z)$ have closed-form expressions. We can then use MCMC to sample from the joint posterior distribution 
\begin{equation}  \label{eq:posterior-joint}
    \pi(\theta, \psi, \beta, \Z\mid\Y) \propto f(\Y \mid\theta, \psi, \beta, \Z) L(\theta,\psi; \Z) \pi(\theta, \psi,\beta).
\end{equation}
Given $M$ MCMC draws $\left\lbrace\left(\theta^{(m)},\psi^{(m)},\beta^{(m)}, \Z^{(m)}\right)\right\rbrace_{m=1}^M$, marginalizing out latent variables \textit{a posteriori} becomes trivial. The values $\left\lbrace\left(\theta^{(m)},\psi^{(m)}, \beta^{(m)}\right)\right\rbrace_{m=1}^M$ form an empirical approximation to the original distribution of interest $\pi\left(\theta,\psi,\beta\mid\Y\right)$. In other words, simply ignoring the latent variables in each posterior sample yields a valid marginal posterior of the parameters of interest. 
Similarly,  $\left\lbrace\Z^{(m)}\right\rbrace_{m=1}^M$ form a Monte Carlo approximation to $\pi\left(\Z\mid\Y\right)$.
An advantage of the Bayesian formulation is that inference is not limited to $\theta$, $\psi$, $\beta$ and $\Z$. The posterior distribution $\pi\left(g(.)\mid\Y\right)$ of any functional $g(\theta, \psi, \beta, \Z)$ is approximated by the transformed sample $\left\lbrace g^{(m)} \right\rbrace_{m=1}^M$, with $g^{(m)} = g\left(\theta^{(m)},\psi^{(m)},\beta^{(m)}, \Z^{(m)}\right)$. In Section~\ref{sec:BCSC}, we will be let $g$ be the overdiagnosis rate.

In our model, the latent variables $\Z$ consist of the unobserved transition age into $P$ and indolent status of each individual. We write $\Z = (\Z_1, \dots, \Z_n)$, where $\Z_i = \left(\Zhpi, \ZIi\right)$ are the latent variables of individual $i$. $\Zhpi \in (t_0, c_i] \cup c_i^+$ is the pre-clinical onset age and $\ZIi\in\{0,1\}$ the indolent status. Here, $c_i^+$ indicates the event that the transition age is larger than $c_i$. Upon introducing these latent variables, we have closed-form expressions for $L(\theta,\psi;\Z)$ and $f(\Y \mid \theta, \psi, \beta, \Z)$. First, by independence, we have
\begin{equation}  \label{eq:likelihood-latent}
    L(\theta,\psi;\Z) = \prod_{i=1}^n L_i(\theta,\psi;\Z_i) = \prod_{i=1}^n \tilde{L}_i(\theta,\psi;\Z_i) / N_i(\theta,\psi, e_i),
\end{equation}
where 
\begin{equation}  \label{eq:likelihood-latent-i}
    \tilde{L}_i(\theta,\psi;\Z_i) \propto f(\Z_i \mid \theta, \psi) = \begin{cases}
        \SH\left( c_i - t_0; \theta\right), & \quad \Zhpi = c_i^+\\
        \fH\left( \Zhpi-t_0; \theta\right) \fI(\ZIi; \psi), & \quad \Zhpi \le c_i
    \end{cases},
\end{equation}
and 
\begin{align} \label{eq:N}
    N_i(\theta,\psi, e_i) := \Pr(\tpci > e_i\mid \theta, \psi)
    & = \Pr(I_i = 1) + \Pr(I_i = 0)\left[\Pr\left(\thpi > e_i \mid I_i = 0\right) + \Pr\left(\thpi < e_i, e_i < \tpci \mid I_i = 0\right) \right]\nonumber \\
    & = \psi + (1-\psi) \left[\SH\left(e_i-t_0;\theta\right) + \int_{t_0}^{e_i} \fH(t-t_0;\theta) \Sprog\left(e_i-t;\theta\right)dt\right]
\end{align}
scales the contribution from person $i$ by the probability that they did not develop clinical cancer before entering the study at age $e_i$. The inclusion of the factor~\eqref{eq:N} in the likelihood~\eqref{eq:likelihood-latent} comes from the fact that data from cohort analyses are left-truncated: participants that develop clinical cancer before the study started are excluded. Therefore, we condition on the fact that the individuals in the sample did not develop clinical cancer before the start of the study. As long as $\fH$ and $\Sprog$ are continuous and bounded, the univariate integral in \eqref{eq:N} is easy to approximate numerically using a quadrature rule (e.g., the function \verb+integrate+ in base \verb+R+).

Second, by independence, we have
\begin{equation}  \label{eq:observation-distribution}
    f(\Y \mid \theta, \psi, \beta, \Z) = \prod_{i=1}^n f(Y_i \mid \Z_i, \theta, \psi, \beta) = \prod_{i=1}^n f_t\left(t^{PC}_i, \delta_i\mid \Z_i, \theta\right) f_O\left(O^i_1,\dots,O^i_{n_i} \mid \Z_i, \beta\right),
\end{equation}
where the distribution of the right-censored clinical onset age is
\begin{equation}  \label{eq:observation-distribution-clinical}
    f_t\left(t^{PC}_i, \delta_i\mid \Z_i, \theta\right) = 
    \begin{cases}
        1-\delta_i, & \quad \Zhpi = c_i^+\\
        \Sprog\left( c_i - \Zhpi; \theta\right)^{1-\ZIi} (1-\delta_i) +  \fprog\left( t^{PC}_i - \Zhpi; \theta\right) \delta_i, & \quad \Zhpi \le c_i
    \end{cases},
\end{equation}
and the distribution of the screen outcomes is
\begin{equation}  \label{eq:observation-distribution-screen}
    f_O\left(O^i_1,\dots,O^i_{n_i} \mid \Z_i, \beta\right) = \prod_{j=1}^{n_i} \left(p^i_j\right)^{O^i_j} \left(1-p^i_j\right)^{1-O^i_j} = \beta^{m_i^+} (1-\beta)^{m_i^-},
\end{equation}
with $m_i^+ = \sum_{j=1}^{n_i} O^i_j$ the number positive screens of person $i$, and $m_i^- = \sum_{j:\Zhpi\le t^i_j} \left(1-O^i_j\right)$ their number of false negative screens.

\subsubsection{Prior specification}  \label{sec:prior}

While the sampling methodology applies to any prior formulation  $\pi(\theta, \psi,\beta)$, we assume for convenience that the parameters are independent \textit{a priori}
\begin{equation}  \label{eq:prior-independent}
    \pi(\theta, \psi, \beta)=\pi(\theta)\pi(\psi)\pi(\beta),
\end{equation}
Moreover, we specify the prior distribution of $\beta$ to be $\pi(\beta)=\Be(a_\beta, b_\beta)$ and that of $\psi$ to be $\pi(\beta)=\Be(a_\beta, b_\beta)$, where $\Be(a,b)$ denotes to the beta distribution with density proportional to $x^{a-1}(1-x)^{b-1}$. The choice of a beta prior distribution for $\beta$ is motivated by conjugacy. Theorem~\ref{theo:conjugacy} holds for any choice of $\FH$ and $\Fprog$; its proof appears in web Appendix~A.

\begin{theorem}[Conjugacy] \label{theo:conjugacy}
If $\beta$ is independent of the other parameters a priori, then the beta distribution is conjugate to the conditional likelihood in~\eqref{eq:posterior-joint} for $\beta$.
In particular, under the prior distribution
\begin{equation} \label{eq:prior-conjugacy} 
   \pi(\theta, \psi, \beta)=\pi(\theta, \psi)\pi(\beta), \quad \pi(\beta)=\Be(a, b),
\end{equation}
the full conditional distribution of $\beta$ is
\begin{equation}  \label{eq:posterior-conjugacy}
    \pi(\beta\mid \theta, \psi, \Z, \Y) = \Be\left(a + m^+, b + m^-\right),
\end{equation}
with $m^+ = \sum_{i=1}^n m_i^+$ and $m^- = \sum_{i=1}^n m_i^-$ the total number of true positive and false negative screens in the sample, respectively.
\end{theorem}


\section{Markov chain Monte Carlo scheme}  \label{sec:mcmc}

We propose a DA-MCMC sampler that alternates between a Gibbs step for $\beta$, univariate Metropolis-Hastings steps for the elements of $\theta$, and for $Z^{HP}_1, \dots, Z^{HP}_n$, and a block Metropolis-Hastings step of $\left(\psi, Z^I_1, \dots, Z^I_n\right)$ \citep{gelman2013bayesian}. This sampling scheme is summarized in Algorithm~\ref{alg:da-mcmc}. The Gibbs step for $\beta$ proceeds by drawing a new value for this parameter from its full conditional distribution~\eqref{eq:posterior-conjugacy} given the most recent values of $\Z$. The proposal distribution for the Metropolis-Hastings steps are described in web Appendix~B. Updating $\psi$ and $\left(Z^I_1, \dots, Z^I_n\right)$ jointly is crucial, because these quantities are highly correlated \textit{a posteriori}. We note that methodologically, the elements of $\theta$ could also be jointly updated in a block Metropolis-Hastings step; though such block step may lead to further efficiency gains, we do not explore this here.

To initialize the DA-MCMC sampler, one only needs to specify the initial values of $(\theta, \psi, \beta)$. Given some initial values $\left(\theta^{(0)}, \psi^{(0)}, \beta^{(0)}\right)$, the algorithm generates the initial values of the latent variables $\Z^{(0)}$ using the independent proposal distributions of the Metropolis-Hastings steps described in web Appendix~B.

\begin{algorithm}
\begin{algorithmic}
\Require a number of iterations $M$, initial values $(\theta^{(0)}, \psi^{(0)}, \beta^{(0)})$, and observed data $\Y$
\State $\Z^{(0)} \leftarrow \text{initialize}\left(\theta^{(0)}, \psi^{(0)}, \beta^{(0)}\right)$
\For{$m = 1, \dots, M$}
    \State $\beta^{(m)} \sim \pi\left(\beta \mid \cdot\right)$ \quad [Gibbs update, c.f. Theorem~\ref{theo:conjugacy}]
    \For{$j = 1, \dots, |\theta|$}
        \State $\theta_j^{(m)}\leftarrow \text{Metropolis-Hastings}\left(\theta_j; .\right)$ \quad [update the $j$th element of $\theta$]
    \EndFor
    \For{$i = 1, \dots, n$}
        \State $\left(\Zhpi\right)^{(m)}\leftarrow \text{Metropolis-Hastings}\left(\Zhpi; .\right)$
    \EndFor
    \State $\psi^{(m)}, \left(Z^I_1, \dots, Z^I_n\right)^{(m)} \leftarrow \text{Metropolis-Hastings}\left(\psi, Z^I_1, \dots, Z^I_n; .\right)$
\EndFor
\State \Return $\left\lbrace\left(\theta^{(m)},\psi^{(m)},\beta^{(m)}, \Z^{(m)}\right)\right\rbrace_{m=1}^M$
\end{algorithmic}
\caption{Steps of the DA-\MCMC sampler.}
\label{alg:da-mcmc}
\end{algorithm}

\subsection{Theoretical guarantees}

Theorem~\ref{theo:erg} shows that the Markov chain underlying the DA-MCMC sampler is ergodic; its proof appears in web Appendix~C. The ergodic theorem, therefore, holds for any $\pi$-integrable function $h$, and the estimator $\bar{h}_M = \frac{1}{M} \sum_{m=1}^{M} h\left(\theta^{(m)}, \psi^{(m)}, \beta^{(m)}, \Z^{(m)}\right)$ is asymptotically consistent for $E_\pi h$ under any initial distribution \citep{billingsley2017probability}.

\begin{theorem}[Ergodicity]  \label{theo:erg}
The Markov chain $\left\lbrace\left(\theta^{(m)}, \psi^{(m)} \beta^{(m)}, \Z^{(m)}\right)\right\rbrace_{m=1}^M$ underlying the DA-MCMC algorithm is ergodic.
\end{theorem}

\subsection{Model comparison}  \label{sec:model-comparison}

Mis-specifying a model of cancer natural history can lead to biased inference \citep{cheung2022multistate}. To identify specifications for $\FH$ and $\Fprog$ that fit the data well, we compare the predictive fit of models with different choices for these distributions. We accomplish this by extending  ALOOCV \citep{gelfand1992model, vehtari2017practical} to the class of latent variable models. While exact CV requires fitting the model multiple times, approximate CV can be easily computed from the posterior draws of a MCMC sampler.
Once we have obtained predictive fit estimates for each model under consideration, we remove the models whose predictive fit is significantly lower than that of the best model, as determined by pairwise t-tests. We then proceed by choosing among the remaining models the one whose specification for $\FH$ and $\Fprog$ we deem agrees the most with existing medical knowledge.
The procedure is straightforward to implement; web Appendix~F contains the complete details of our extension of ALOOCV to latent models.

\subsection{Model implementation}  \label{sec:model-implementation}

The methodology described in Section~\ref{sec:bayesian-inference} applies to any continuous distribution $\FH$ and $\Fprog$ on the positive line. For the simulations in Section~\ref{sec:sim} and the case study in Section~\ref{sec:BCSC}, we choose $\FH=\Wei(\lH, \aH)$ and $\Fprog=\Wei(\lprog,\aprog)$, with $\Wei(\lambda, \alpha)$ denoting the Weibull distribution with rate $\lambda>0$ and shape $\alpha>0$, and whose density is proportional to $x^{\alpha-1}\exp\{-\lambda x^{\alpha}\}$. The Weibull choice has several advantages, including flexibility and interpretability: if $\alpha>1$ ($\alpha<1)$, the hazard function increases (decreases) with age as $\propto t^{\alpha-1}$; and if $\alpha=1$ it corresponds to the exponential distribution. In addition, the Weibull distribution's survival function has a closed-form expression, and its cumulative distribution function can be inverted, enabling efficient sampling from $\trfh$ in the proposal; distribution described in web Appendix~B.`

We implement the MCMC sampler in the programming language \verb|C++| for efficiency. The sampler is fast, scaling to cohort studies with tens of thousands of individuals. For convenience, the \verb|C++| implementation is made available in \verb|R| with the help the package \verb+RCPP+ \citep{eddelbuettel2011rcpp} via the \verb|R| package \verb|baclava| available on \verb|CRAN|. In the MCMC runs in Sections~\ref{sec:sim}-\ref{sec:BCSC}, we keep $10^3$ thinned draws post warm-up phase due to memory constraints.

\section{Performance on simulated data}  \label{sec:sim}

To examine the mixing properties of our sampler, we run multiple MCMC samplers initialized at over-dispersed values on a simulated dataset. We then use the potential scale reduction factor \citep{gelman1992inference, brooks1998general}, the effective sample size, and traceplots to assess the algorithm's mixing properties.

For the purposes of a simulated dataset, we set the starting age for non-zero onset hazard to $t_0=30$ years, and assume a linear increase in both the onset and progression hazards ($\alpha_H=\alpha_{prog}=2$). We further set the Weibull rates to $\lH=6.5\times10^{-5}$ and $\lprog=3.14\times10^{-2}$, resulting in a cumulative risk of pre-clinical onset by age $80$ of $15\%$, and a mean sojourn time of 
$\muprog=5$ years. Finally, we set the indolent probability to $\psi=0.1$ and the screen sensitivity to $\beta=0.85$. We simulate the natural histories of cancer among $n=40,000$ individuals under the following distribution of first-screen age, inter-screen interval and age of last follow-up. The age of the first screen is sampled from the categorical distribution with weights $w_k \propto \exp\{-k/5\}$ for $k=40,41,\dots,80$. Inter-screen intervals are sampled from the distribution $\textsf{Poisson}(0.5)$ plus $1$ year, resulting in a mean inter-screen interval of $1.5$ years. Finally, for individual $\iton$, the age of last follow-up is $\min\{100, t^i_1 + W_i\}$ with $W_i \sim \Exp(1/5)$ independently over $i$, implying that individuals remain in the screening program for $5$ years on average.

We use the independent prior distribution $\pi(\lH, \lprog, \psi, \beta) = \pi(\lH)\pi(\lprog)\pi(\psi)\pi(\beta)$ with 
\begin{equation*}  
    \pi(\lH)=\Ga(1, 0.01), \quad \pi(\lprog)=\Ga(1, 0.01), \quad \pi(\psi) = \Be(1, 1), \quad \pi(\beta)=\Be(38.5, 5.8),
\end{equation*}
where $\Ga(a,b)$ corresponds to the gamma distribution with density proportional to $x^{a-1}\exp\{-bx\}$.
We choose these distributions for their flexibility and, in the case of $\pi(\beta)$, for its conjugacy (Theorem~\ref{theo:conjugacy}). The prior distribution on $\lH$, $\lprog$ and $\psi$ are weakly informative, and that of $\beta$ places $95\%$ of its mass on the interval $(0.76, 0.95)$, reflecting prior empirical data on the performance of mammography screening \citep{Lehman2017national}.

Next, we apply our MCMC sampler to the simulated dataset and estimate the values of the parameters $\lH$, $\lprog$, $\psi$ and $\beta$. We independently run $20$ MCMC samplers initialized at over-dispersed values of these parameters (the latent variables need not be initialized, see Section~\ref{sec:mcmc}). The initial values $\lH^{(0,m)}$, $\lprog^{(0,m)}$, $\psi^{(0,m)}$ and $\beta^{(0,m)}$ of the $m$th sampler $(m=1\dots,20)$ are sampled from a wide interval around the true values, mutually independently,
\begin{align*}
    & 10^4\times \lH^{(0,m)} \sim \Unif(0.2, 2), \quad \lprog^{(0,m)} \sim \Unif(0.01,0.1), \quad
    \psi^{(0,m)} \sim \Unif(0,1), \quad \beta^{(0,m)} \sim \Unif(0.7, 0.95).
\end{align*}

For each chain, we set the warm-up period to $10^4$ iterations and run the sampler for an additional $2\times10^4$ iterations. For each parameter, the potential scale reduction factor, computed on the post warm-up thinned draws with the \verb+R+ package \verb+coda+ \citep{plummer2006coda}, is smaller than $1.005$ and the effective sample size is superior to 2,500, indicating that the chains---which were initialized at over-dispersed values---converged during the warmup phase and mixed well once they reached stationarity \citep{gelman1992inference}. This is a remarkable achievement given the high dimensional latent space of 80,000 latent variables.


In addition to exploring the parameter space efficiently, the sampler also explores the space of the latent variables swiftly. Figure~\ref{fig:mcmc-latent} displays the MCMC draws of the latent variables for an individual whose cancer was screen-detected at age $58$ years, after a series of five consecutive negative screens.  The individual had $5$ negative screens at ages $46$, $48$, $51$, $53$ and $55$. The traceplots of the indolent indicator (Figure~\ref{fig:traceplot-indolent}) and the pre-clinical transition time (Figure~\ref{fig:traceplot-tau}) indicate that the sampler rapidly explores the latent space. In particular, the frequent switches between the values $0$ and $1$ of the binary draws in Figure~\ref{fig:traceplot-indolent} show that the Markov chain is not stuck in any particular configuration of the latent variables for any large number of iterations.

An interesting feature of our data augmentation approach is that we obtain a patient-specific estimate of the onset time of pre-clinical cancer, a quantity unobservable in practice. Figure~\ref{fig:histogram-tau} illustrates the MCMC approximation of the marginal posterior distribution $\pi\left(\tau^{HP}\mid\Y\right)$ for the above individual with a screen-detected cancer at age 58 years. Because the screening sensitivity is estimated to be between $0.75$ and $0.90$, it is most likely that the time of onset took place  after the last screen at $55$ years, less likely that it did between the screens at ages $53$ and $55$, and very unlikely before the screen at age $53$. As illustrated by the discontinuities in this histogram, the screens themselves contribute most of the information available about the unobserved onset age of pre-clinical cancer.

\begin{figure}
    \centering
    \includegraphics[width=1\linewidth]{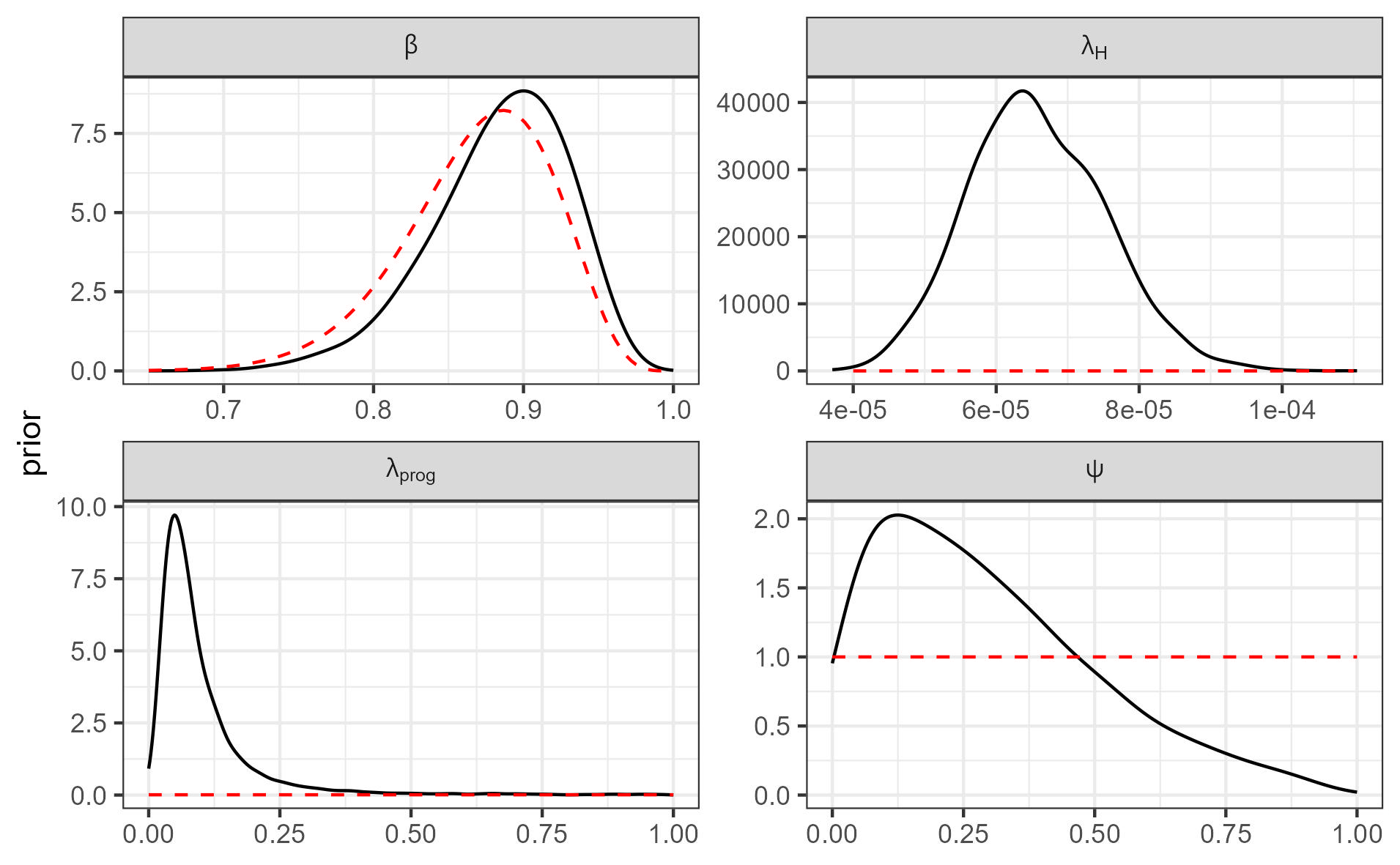}
    \caption{Estimate of the density of the posterior distribution (solid black line) based on the MCMC draws and density of the prior distribution (dashed blue line) for the simulated data.}
    \label{fig:mcmc-parameters}
\end{figure}

\begin{figure}
\centering
\begin{subfigure}{0.3\textwidth}
    \includegraphics[width=\textwidth]{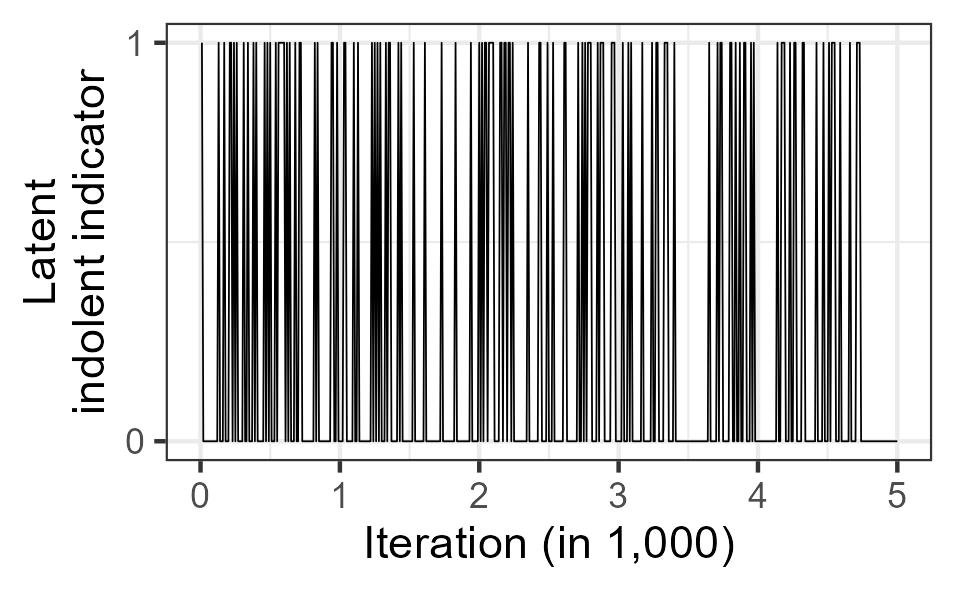}
    \caption{Traceplot of $Z^I$.}
    \label{fig:traceplot-indolent}
\end{subfigure}
\hfill   
\begin{subfigure}{0.3\textwidth}
    \includegraphics[width=\textwidth]{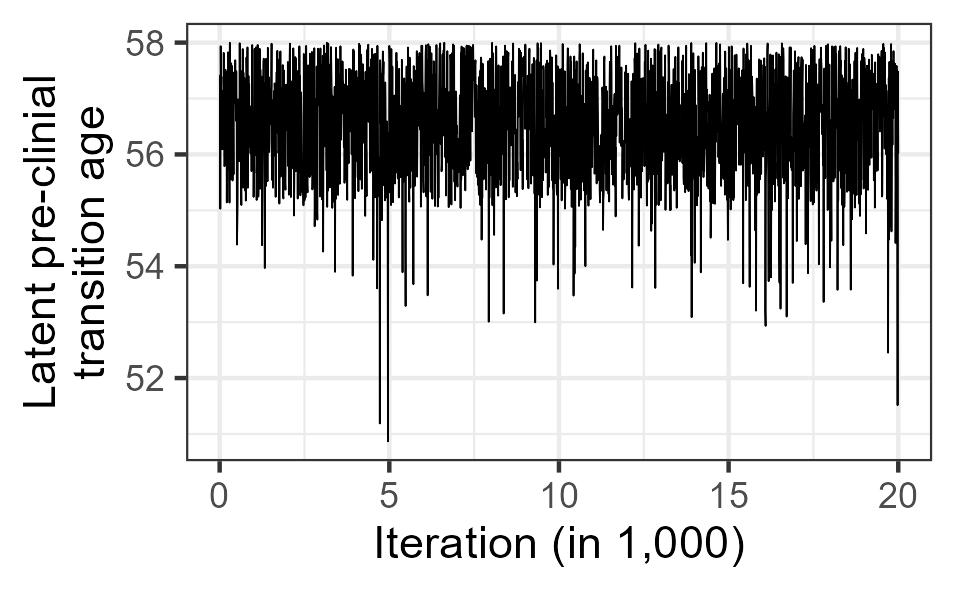}
    \caption{Traceplot of $Z^{HP}$.}
    \label{fig:traceplot-tau}
\end{subfigure}
\hfill
\begin{subfigure}{0.3\textwidth}
    \includegraphics[width=\textwidth]{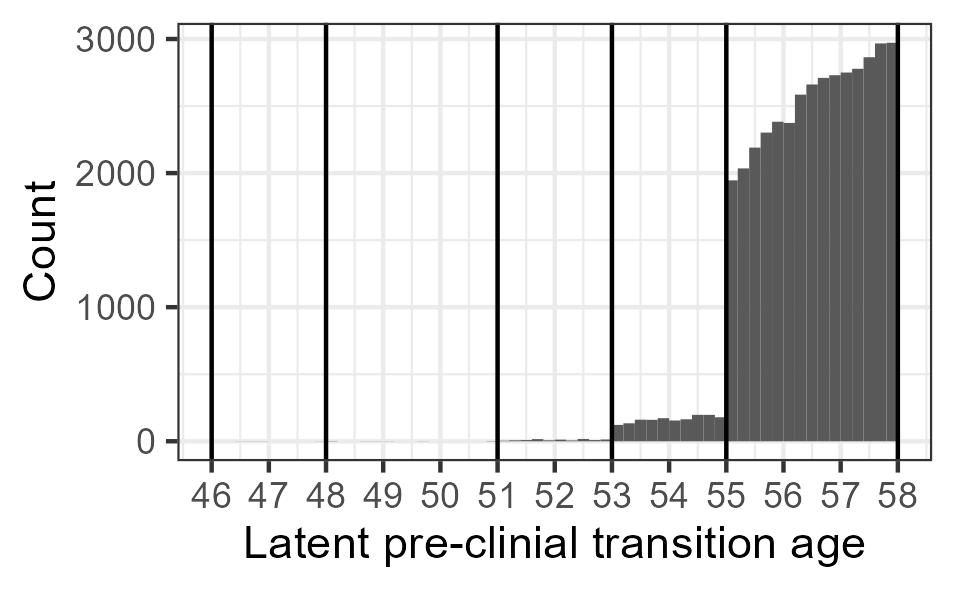}
    \caption{Histogram of $Z^{HP}$.}
    \label{fig:histogram-tau}
\end{subfigure}    
     
\caption{MCMC draws from the first chain ($m=1$) for the latent variables of an individual with a screen-detected cancer at age $58$ and five negative screens at ages $(46, 48, 51, 53, 55)$. Panel (a) is a traceplot of the latent indolent indicator over an interval of $5.000$ iterations for clarity. Panel (b) is a traceplot of the latent pre-clinical transition age. Panel (c) shows the marginal posterior distribution of the latent pre-clinical transition age; vertical lines depict the ages at which the screens occur.}
\label{fig:mcmc-latent}
\end{figure}

\section{Application to data from the \BCSC} \label{sec:BCSC}

We illustrate the utility of our semi-Markov model and accompanying DA-MCMC sampler with an analysis of the screening and cancer diagnosis histories from participants in the \BCSC (BCSC), a large population-based mammography screening registry in the US \citep{ballard1997breast}. The BCSC registries collect mammography results and link them to breast cancer outcomes and vital status through linkage  with regional population-based  Surveillance, Epidemiology, and End Results (SEER) programs, state tumor registries, and state death records.

We use the analytic cohort previously described in \citep{ryser2022estimation}, which includes $35,986$ women aged $50$ to $74$ years who received their first screening mammogram at a BCSC facility between $2000$ and $2018$. Details are provided in Wed Appendix G. Each BCSC registry and the Statistical Coordinating Center (SCC) have received institutional review board approval for all study procedures, including passive permission (one registry), a waiver of consent (six registries), or both depending on facility (one registry), to enroll participants, link data, and perform analytic studies. All procedures are Health Insurance Portability and Accountability Act (HIPAA) compliant. All registries and the SCC have received a Federal Certificate of Confidentiality and other protections for the identities of women, physicians, and facilities who are subjects of this research.

Following \citep{ryser2022estimation}, we set $t_0=30$.
We employ the independent prior $\pi(\lH, \lprog, \psi, \beta) = \pi(\lH)\pi(\lprog)\pi(\psi)\pi(\beta)$ with 
\begin{equation*}  
    \pi(\lH)=\Ga(1, 0.01), \quad \pi(\lprog)=\Ga(a_\textsf{prog}, b_\textsf{prog}), \quad \pi(\psi) = \Be(1, 1), \quad \pi(\beta)=\Be(38.5, 5.8),
\end{equation*}
where, for a given $\aprog$, we choose values for $a_\textsf{prog}$ and $b_\textsf{prog}$ that induce a weakly informative distribution on the mean sojourn time $\muprog$ with $0.9$ mass on the interval $(1,9)$ years, reflecting prior knowledge about breast cancer progression. Web Appendix~I contains the details.

\subsection{Model specification}

We consider the $21$ models with $\aH\in\{1,1.5,2,2.5,3,4,5\}$ and $\aprog\in\{1,1.5,2\}$. We fit each model to the BCSC data and estimate their predictive fit via ALOOCV (see Section~\ref{sec:model-comparison}). As no cancer screens in the data are conducted before age $50$, we set $t_0=30$. The burn-in period is set to $2\times10^4$ iterations and is followed by an additional $5\times10^4$ iterations; this results in an effective sample size of at least $10^3$ for each parameter in every setting.

Figure~\ref{fig:aloocv} displays the predictive fit of the models. The relation between the predictive fit and $\aH$ follows an inverted-U shape, and, nominally, the best model fit was achieved at $\aH=4$ and $\aprog=1$. Models with $\aH\le1.5$ and $2$ of the $3$ models with $\aH=2$ have a predictive fit that is significantly worse than the best model at the significance level $0.05$. This suggests that the pre-clinical hazard rate of breast cancer increase at least linearly with age. However, the exact speed of the increase---linear, quadratic, etc---cannot be inferred from these data alone, presumably because the available screens only cover the range $50$-$85$ years, with $97\%$ of them limited to the range $50$-$75$ years.

\begin{figure}
    \centering
    \includegraphics[width=0.5\linewidth]{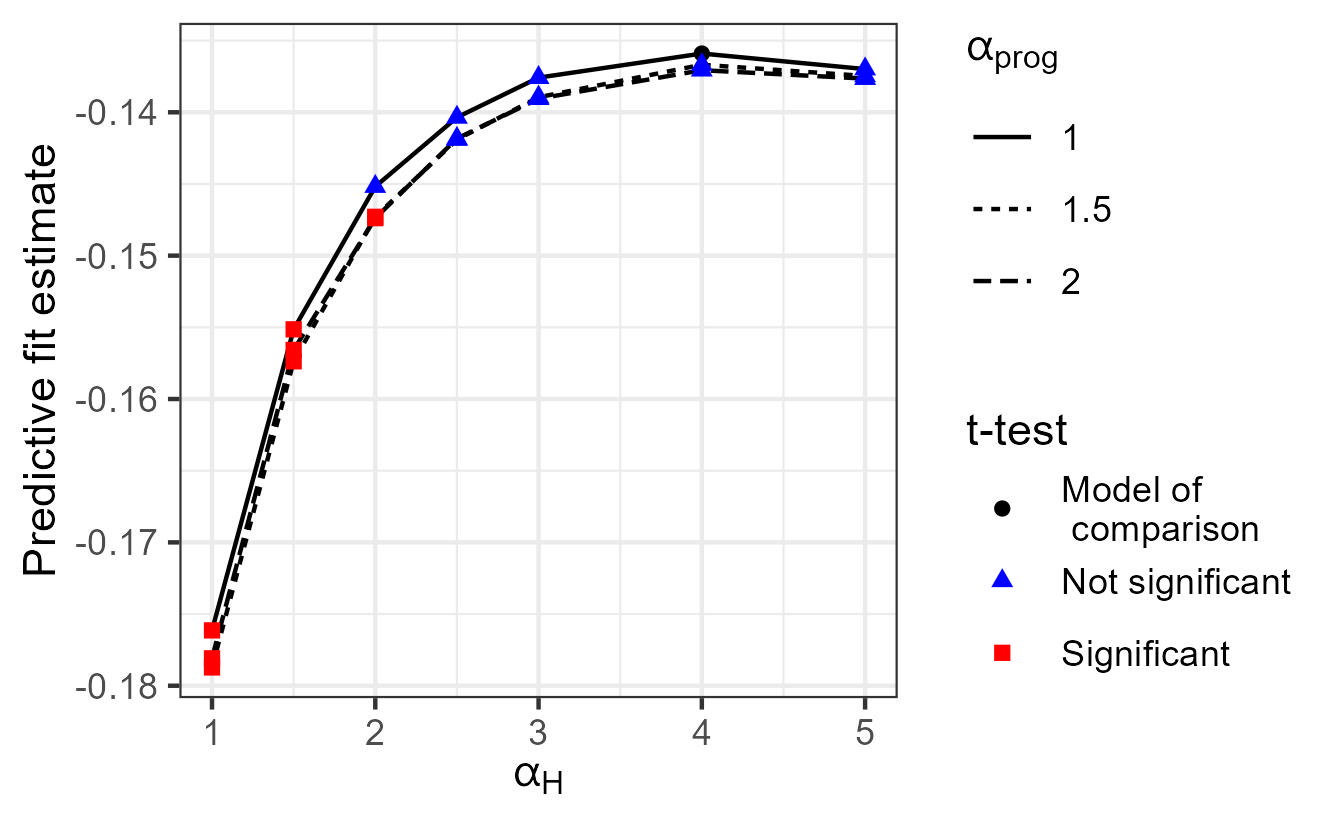}
    \caption{Predictive fit of the models on the BCSC data. The black dot denotes the best model. The red squares correspond to models whose predictive score is significantly different from that of the best model, based on a two-sample t-test at the significance level $0.05$). The blue triangle correspond to models whose predictive fit is not significantly different from that of the best model.}
    \label{fig:aloocv}
\end{figure}

In contrast, $\aprog$ does not have a significant impact on the predictive fit. The three lines in Figure~\ref{fig:aloocv}, which correspond to the three values of this parameter, are extremely close to one another, and the pairwise differences between the predictive fit of models with the same $\aH$ did not reach statistical significance at $\alpha=0.05$. This indicates that the data do not favor any particular value for this parameter. This is to be expected because the endpoints of the pre-clinical sojourn times of progressive cancer are mostly unobserved. Indeed, the start age is never exactly observed; furthermore, as the screens are imperfect, this value is not even interval censored. In addition, the end age of the sojourn is exactly observed only for the interval-detected group ($0.18\%$ of the sample) and is right-censored for the other individuals. This makes the task of learning the exact shape of the distribution of sojourn times from the observed data alone extremely difficult.

We now proceed by choosing among the models whose predictive fit is not significantly worse than the best model the one with values for $\aH$ and $\aprog$ that we deem most plausible given our medical understanding of the progression of breast cancer. First, we choose a value of $\aprog$ that implies a realistic predictive distribution of sojourn times among progressive cancers. Conditional on $\lprog$, the sojourn time of a progressive cancer---denoted by $\sigma_\textsf{prog}$---follows the Weibull distribution $\Wei\left(\lprog, \aprog\right)$. Its predictive distribution
\begin{equation}  \label{eq:predictive}
    \pi(\sigma_\textsf{prog}\mid\Y, \aprog) = \int\Wei\left(\sigma_\textsf{prog};\lprog, \aprog\right)\pi\left(\lprog\mid\Y, \aprog\right)d\lprog,
\end{equation}
is therefore a rate-mixture of Weibull distributions, in which the Weibull rate is integrated out. We approximate the posterior distribution $\pi\left(\lprog\mid\Y, \aprog\right)$ with the empirical distribution of the MCMC draws $\left\lbrace\lprog^{(1)}, \dots, \lprog^{(M)}\right\rbrace$, giving the Rao-Blackwell estimate of the predictive distribution~\eqref{eq:predictive} \citep{gelman2013bayesian}
\begin{equation}  \label{eq:rao-blackwell}
    \hat{\pi}(\sigma_\textsf{prog}\mid\Y,\aprog) = \dfrac{1}{M} \sum_{m=1}^M \Wei\left(\sigma_\textsf{prog};\lprog^{(m)}, \aprog\right).
\end{equation}
Web Appendix~H compares the estimate~\eqref{eq:rao-blackwell} for the models with $\aH=2$ and $\aprog\in\{1,1.5,2\}$. Under $\aprog=1$, the predictive distribution has a large amount of  mass concentrated at $0$ and a fat tail. In contrast, under the semi-Markov models with $\aprog\in\{1.5,2\}$ the density of the predictive distribution has a thin tail, and no mass concentrated at $0$. As a result, the predictive probability of a sojourn time less than $0.5$ year or more than $15$ years is almost $20\%$ under $\aprog=1$, while it is only $1.3\%$ under $\aprog = 2$. A similar patterns is observed across all values of $\aH$. We thus opt for $\aprog=2$ because it results in a more realistic model of sojourn times.

Next, we choose the value of $\aH$. We opt for $\aH=2.5$ because it agrees the most with the trend in age-specific data from SEER. Figure~\ref{fig:mcmc-bcsc} presents the posterior distribution of key parameters under the chosen model. The posterior mean of the indolence probability $\psi$ is $0.057$, with $95\%$ credible interval ($0.002-0.17$), and that of the mean sojourn time $\muprog$ is $5.6$ years ($95\%$ CI: $4.1-7.1$). Finally, the cumulative probability of pre-clinical onset by age 80 years is estimated to be $12.3\%$ ($95\%$ CI: $11.2\%-13.4\%$), which is consistent with an approximate 1 in 7 lifetime risk of developing breast cancer \citep{howlader2014seer}.

\begin{figure}
\centering
\begin{subfigure}{0.3\textwidth}
    \includegraphics[width=\textwidth]{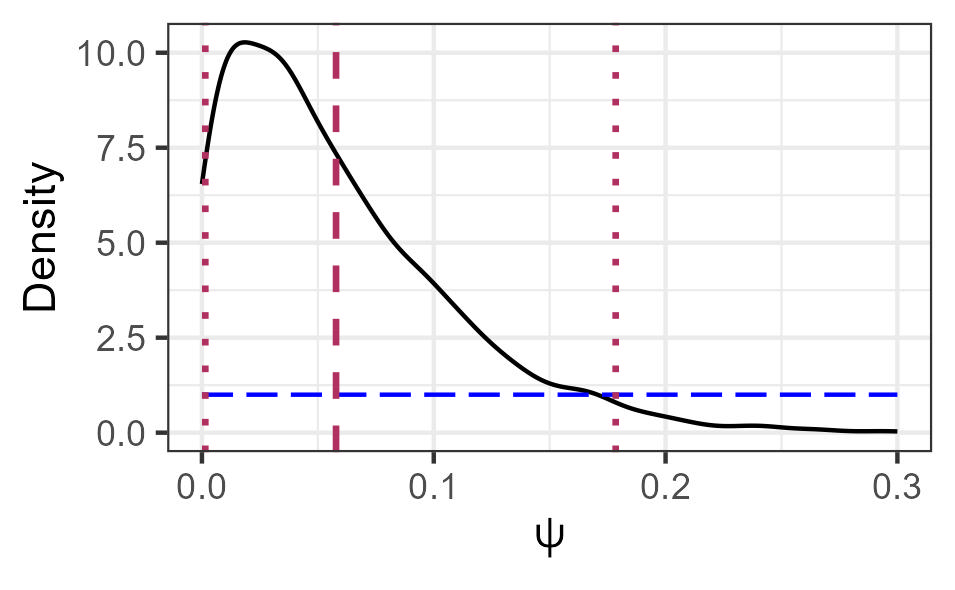}
    \caption{Posterior distribution of $\psi$.}
    \label{fig:histogram-psi}
\end{subfigure}
\hfill
\begin{subfigure}{0.3\textwidth}
    \includegraphics[width=\textwidth]{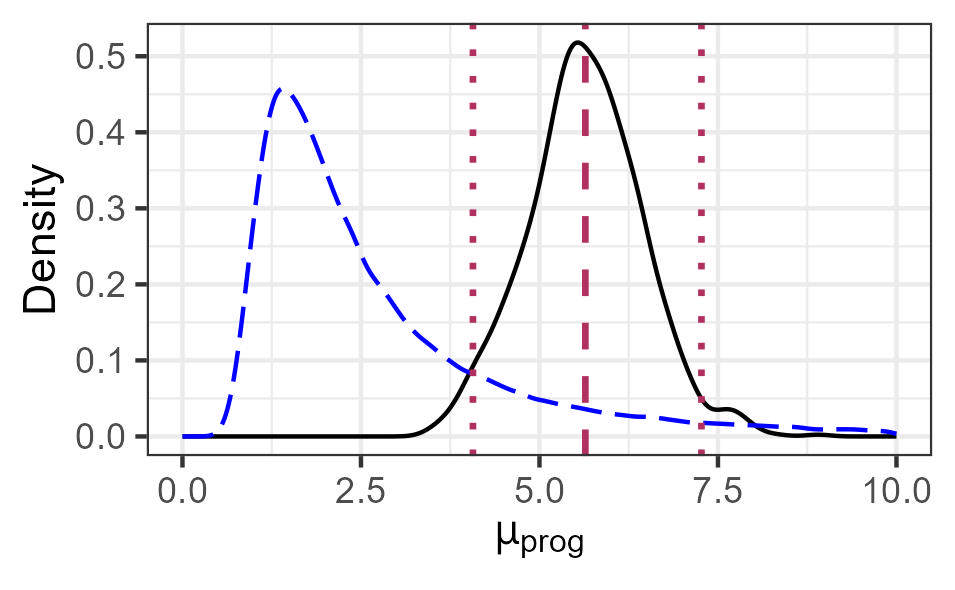}
    \caption{Posterior distribution of $\muprog$.}
    \label{fig:histogram-mu_P}
\end{subfigure}
\hfill
\begin{subfigure}{0.32\textwidth}
    \includegraphics[width=0.95\textwidth]{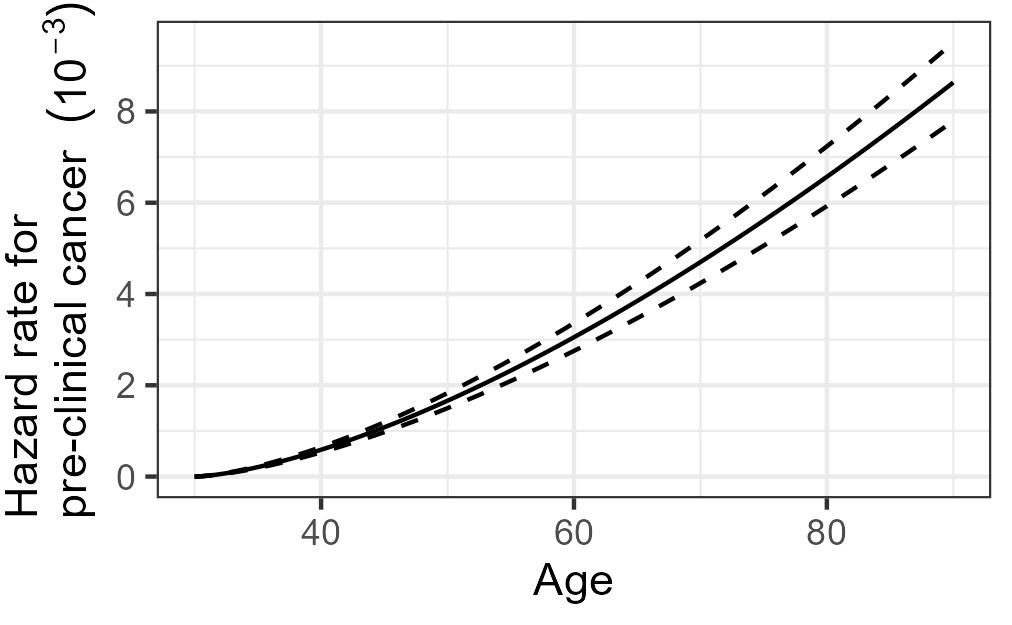}
    \caption{Estimated pre-clinical hazard function.}
    \label{fig:hazard}
\end{subfigure}
        
\caption{Results of the analysis of the data set from the \BCSC data set. Panels (a) and (b) display the marginal posterior distributions of the fraction of indolent pre-clinical cancers ($\psi$) and the mean sojourn time among progressive pre-clinical cancers ($\muprog$), respectively. The posterior means (dashed red line) and $95$\% credible intervals (dotted red lines) are included. The dashed blue lines correspond to the prior distributions.
Panel (c) shows the estimated pre-clinical hazard rate as a function of age, with the posterior median (solid line) and $95$\% credible band (dashed lines) indicated.
}
\label{fig:mcmc-bcsc}
\end{figure}

\subsection{Estimation of the overdiagnosis rate}  \label{sec:bcsc-overdiagnosis}

From the MCMC draws, one can estimate the posterior distribution of any function of the parameters. Here, we focus on the extent of \textit{overdiagnosis} in a screening program, defined as the the mammographic detection of cancer that would not become clinically evident in the woman's remaining lifetime. There are two contributions to overdiagnosis. First, by definition, every indolent cancer detected on mammography is overdiagnosed. Second, progressive screen-detected cancer are overdiagnosed if the time from screen-detection to death from a breast cancer-unrelated cause is shorter than the time from screen-detection to clinical progression (also referred to as the lead-time). 

We use the breast cancer natural history estimates from the BCSC cohort together with life tables that exclude death from breast cancer to estimate overdiagnosis rates, overall and by source, in a biennial screening program from age $50$ to $74$ years; details of the calculation are provided in web Appendix~J. Table~\ref{tab:overdiagnosis} provides a side-by-side comparison of the three models with $\aH=2.5$.  The three models yield similar estimates for the overdiagnosis rate and the respective contributions from indolence and mortality from other sources, showing that these results are robust to the choice of $\aprog$. This is reassuring given that the value of this parameter could not be learned from the data. The overdiagnosis estimates are similarly robust to values of $\aH$ around $2.5$. Moreover, the estimates we obtain are in line with those in \citep{ryser2022estimation}, which are based on a Markov model with a piece-wise constant onset rate for pre-clinical cancer, corroborating the claim that approximately 1 in 7 screen-detected breast cancers are overdiagnosed \cite{ryser2022estimation}.

\begin{center}
\begin{table*}[!h]%
\caption{Estimates of the overdiagnosis rate and the respective contribution of indolence and mortality from other causes for the nine models with $(\aH, \aprog) = \{2, 2.5, 3\}\times\{1, 1.5, 2\}$ and the model in \citep{ryser2022estimation}. The posterior means and $95\%$ credible intervals are reported.\label{tab:overdiagnosis}}
\begin{tabular*}{\textwidth}{@{\extracolsep\fill}llccccc@{}}
\toprule
		$\mathbf{\aH}$ & $\mathbf{\aprog}$ & \textbf{\shortstack{Overdiagnosis\\rate (\%)}} & \multicolumn{2}{@{}l}{\textbf{Source}} & \textbf{\shortstack{Mean sojourn\\time}} & \textbf{\shortstack{Probability of\\indolence$^{a}$ $(\%)$}}\\
		\cmidrule(r){4-5}
        & & & \textbf{Indolence (\%)}     & \textbf{\shortstack{Mortality from\\other causes (\%)}} & & \\
        \midrule
		$2$   & $1$   & $18.8$ $(6.5, 35.9)$ & $8.1$  $(0.2, 28.0)$ & $10.7$ $(2.8, 23.0)$ & $6.6$ $(5.0, 8.7)$ & $4.3$  $(0.1, 13.1)$\\
		$2$   & $1.5$ & $16.9$ $(7.9, 36.5)$ & $12.7$ $(3.8, 35.7)$ & $4.1$  $(0.7, 10.0)$ & $5.6$ $(4.0, 7.1)$ & $5.7$  $(0.2, 15.6)$\\
		$2$   & $2$   & $16.9$ $(5.6, 33.1)$ & $11.7$ $(1.0, 31.4)$ & $5.2$  $(0.9, 13.2)$ & $5.1$ $(3.4, 7.1)$ & $9.0$  $(0.5, 21.8)$\\
		$2.5$ & $1$   & $14.0$ $(4.1, 26.5)$ & $7.3$  $(0.2, 21.7)$ & $6.7$  $(1.4, 16.2)$ & $7.0$ $(5.3, 9.0)$ & $4.3$  $(0.1, 14.8)$\\
		$2.5$ & $1.5$ & $13.2$ $(3.9, 26.0)$ & $5.9$  $(0.2, 20.0)$ & $7.3$  $(1.6, 17.1)$ & $5.8$ $(4.4, 7.4)$ & $4.8$  $(0.2, 15.5)$\\
		$2.5$ & $2$   & $12.2$ $(3.6, 23.8)$ & $4.7$  $(0.1, 16.4)$ & $7.5$  $(1.8, 17.5)$ & $5.6$ $(4.1, 7.3)$ & $5.8$  $(0.1, 17.8)$\\
		$3$   & $1$   & $11.8$ $(3.4, 23.6)$ & $4.5$  $(0.1, 15.5)$ & $7.3$  $(1.8, 16.8)$ & $7.2$ $(5.5, 9.2)$ & $4.3$  $(0.1, 14.7)$\\
		$3$   & $1.5$ & $12.0$ $(3.2, 25.4)$ & $5.1$  $(0.1, 19.6)$ & $6.9$  $(1.6, 16.0)$ & $6.0$ $(4.6, 7.5)$ & $4.0$  $(0.1, 13.0)$\\
		$3$   & $2$   & $18.9$ $(9.4, 42.1)$ & $15.2$ $(5.4, 41.5)$ & $3.8$  $(0.6, 9.2)$  & $5.7$ $(4.3, 7.0)$ & $4.8$  $(0.2, 16.4)$\\
		\multicolumn{2}{c}{Model in\citep{ryser2022estimation}} 
                      & $15.4$ $(9.4,26.5)$  & $6.1$ $(0.2,20.1)$ & $9.3$ $(5.5,13.5)$ & $6.6$ $(4.9, 8.6)$ & $4.5$ $(0.1, 14.8)$\\
		\bottomrule
\end{tabular*}
\begin{tablenotes}
\item[$^{\rm a}$] Proportion of pre-clinical cancers that are indolent.
\end{tablenotes}
\end{table*}
\end{center}

\section{Discussion}  \label{sec:discussion}

This paper introduces a data-augmented MCMC sampler for fitting semi-Markov mixture models of cancer natural history to individual-level screening and cancer diagnosis histories. By allowing for age-dependent pre-clinical onset hazards, non-exponential sojourn times, and a mixture of indolent and progressive tumors, we substantially expand the range of cancer progression models and dynamics amenable to rigorous Bayesian inference. These advances mark a departure from prior approaches relying on simple Markov models for tractability. Our sampler efficiently explores the high-dimensional joint space of model parameters and latent variables, and automatically quantifies individual-level uncertainty of clinically relevant quantities such as a person's pre-clinical onset time, indolence status, and sojourn time.

A key application of our methodology is the estimation of screening-related overdiagnosis rates. In the motivating example of breast cancer, overdiagnosis has been documented as the chief harm of screening, yet there remains substantial uncertainty about its true extent. Due to the counterfactual nature of overdiagnosis, its estimation is methodologically challenging and, as in our case of observational data from a screened cohort, relies on sound estimates of the latent tumor progression dynamics. Applying our method to a high-quality cohort of women undergoing mammography screening, we make two important observations. First, our breast cancer overdiagnosis estimates are similar to those previously derived from the same cohort using a slightly different model structure, thus adding weight to the finding that approximately 1 in 7 screen-detected breast cancers are overdiagnosed. Second, although the data contain insufficient information to delineate the best-fitting sojourn time parameterization, we find that the overdiagnosis estimates were robust to changes in this parameterization. Robustness is particularly important because model misspecification is inevitable when modeling complex biological processes such as cancer progression. Finally, the latent variables provide individual-level estimates of pre-clinical cancer onset times. These data could, for example, be leveraged to derive individualized estimation of overdiagnosis.

The methodology presented here has focused on newly enabling exact inference beyond the Markovian setting. While we show that this increased model flexibility remains tractable for inference in large datasets, we note several limitations that may continue to be addressed in future work. Age-dependent extensions will allow for remaining restrictive assumptions to be relaxed. For example, allowing the sojourn time $\sigma_P$ to depend on the age of pre-clinical onset is more consistent with clinical knowledge that cancers with later onset tend to progress more slowly. While the method in principle can tackle semi-Markovian models in broad generality, its advantages are in part hampered by fixing a subset of the parameters. Our empirical results via ALOOCV suggest  that incorporating parameters such as the onset shape parameter $\alpha_H$ into the DA-MCMC sampler may be practically identifiable in some data settings.  Similarly, including patient-level covariates into model will further improve model fit and realism, and may find use toward personalized prevention strategies. We note that likelihood-free approaches are readily applicable to include such features and more complex parametrizations, albeit at the cost of intensive simulation-based computation. Structured parametrizations such as log-linear relationships may remain amenable to an exact sampling approach, and have also been studied using EM algorithms \citep{bu2025stochastic}. Future extensions of the method and software targeting such covariate-dependent rates, as well as   nonparametric specifications of the hazard, are promising directions to expand the contributions in this article.


\bmsection*{Author contributions}

The first author developed the DA-MCMC sampler, adapted ALOOCV to latent models, conducted the simulations in Section~\ref{sec:sim} and the case study in Section~\ref{sec:BCSC}, and wrote the manuscript. The second author implemented the DA-MCMC sampler and the ALOOCV adaptation in \verb|C++| and made this code available in the \verb|R| package \verb|baclava| on \verb|CRAN|; the third and fourth author are joint senior authors and provided conception of methodologies, statistical guidance, and domain expertise. 

\bmsection*{Acknowledgments}
This research was funded by the National Cancer Institute (R35-CA274442) and the National Science Foundation (DMS 2230074 and PIPP 2200047). Data collection for this research was funded by the National Cancer Institute (P01CA154292, U01CA63736, U01CA069976). The collection of cancer and vital status data was supported in part by several state public health departments and cancer registries (https://www.bcsc-research.org/work/acknowledgement). We thank the participating women, mammography facilities, and radiologists for the data they have provided. You can learn more about the BCSC at: http://www.bcsc-research.org/. All statements in this report, including its findings and conclusions, are solely those of the authors and do not necessarily represent the views of the National Cancer Institute or the National Institutes of Health.

\bmsection*{Financial disclosure}
None reported.

\bmsection*{Conflict of interest}
The authors declare no potential conflict of interests.

\bibliography{main-ref}

\bmsection*{Supporting information}

Additional supporting information may be found in the
online version of the article at the publisher’s website.

\end{document}